\let\latexdocument\document
\let\latexarabic\arabic

\documentclass[article,lineno]{biometrika}

\let\document\latexdocument
\let\arabic\latexarabic

\usepackage{setspace}

\makeatletter

\copyrightinfo{}  

\def\printhistory{}

\let\biom@origps@plain\ps@plain
\def\ps@plain{%
  \biom@origps@plain
  \def\@oddhead{\smash{\vbox to 6pt{%
    \hsize\textwidth
    \slugfont
    \hfill\thepage\hfill
  }}}
  \def\@evenfoot{}%
  \def\@oddfoot{}
  \let\@evenhead\@oddhead
}

\makeatother

\usepackage{graphicx}%
\usepackage{multirow}%
\usepackage{amsmath,amssymb,amsfonts}%
\usepackage[title]{appendix}%
\usepackage{xcolor}%
\usepackage{textcomp}%
\usepackage{manyfoot}%
\usepackage{booktabs}%
\usepackage{listings}%
\usepackage{url}
\usepackage [latin1]{inputenc}

\usepackage{newtxtext}
\usepackage[subscriptcorrection]{newtxmath}

\graphicspath{{./art/}}

\usepackage[plain,noend]{algorithm2e}

\makeatletter
\renewcommand{\algocf@captiontext}[2]{#1\algocf@typo. \AlCapFnt{}#2} 
\def\@algocf@capt@plain{top}
\renewcommand{\algocf@makecaption}[2]{%
  \addtolength{\hsize}{\algomargin}%
  \sbox\@tempboxa{\algocf@captiontext{#1}{#2}}%
  \ifdim\wd\@tempboxa >\hsize
    \hskip .5\algomargin%
    \parbox[t]{\hsize}{\algocf@captiontext{#1}{#2}}
  \else%
    \global\@minipagefalse%
    \hbox to\hsize{\box\@tempboxa}
  \fi%
  \addtolength{\hsize}{-\algomargin}%
}
\makeatother

\begin{document}

\jname{}
\jyear{}
\jvol{}
\jnum{}
\cyear{}
\accessdate{}

\received{}
\revised{}

\markboth{E.G. Crenshaw et~al.}{}

\title{Covariate Selection for Joint Latent Space Modeling of Sparse Network Data}

\author{E.G. Crenshaw}
\affil{Department of Biostatistics, Harvard TH Chan School of Public Health,\\ 677 Huntington Avenue, Boston, MA 02115.
\email{emma\_crenshaw@g.harvard.edu}}

\author{Yuhua Zhang}
\affil{Department of Biostatistics, Harvard TH Chan School of Public Health,\\ 677 Huntington Avenue, Boston, MA 02115. \email{yuhuazhang@hsph.harvard.edu}}

\author{Jukka-Pekka Onnela}
\affil{Department of Biostatistics, Harvard TH Chan School of Public Health,\\ 677 Huntington Avenue, Boston, MA 02115.
\email{onnela@hsph.harvard.edu}}

\maketitle

\raggedbottom



\begin{abstract}
    Network data are increasingly common in the social sciences and infectious disease epidemiology. Analyses often link network structure to node-level covariates, but existing methods falter with sparse networks and high-dimensional node features. We propose a joint latent space modeling framework for sparse networks with high-dimensional binary node covariates that performs covariate selection while accounting for uncertainty in estimated latent positions. Building on joint latent space models that couple edges and node variables through shared latent positions, we introduce a group lasso screening step and incorporate a measurement-error-aware stabilization term to mitigate bias from using estimated latent positions as predictors. We establish prediction error rates for the covariate component both when latent positions are treated as observed and when they are estimated with bounded error; under uniform control across $q$ covariates and $n$ nodes, the rate is of order $O(\log q / n)$ up to an additional term due to latent position estimation error. Our method addresses three challenges: (1) incorporating information from isolated nodes, which are common in sparse networks but often ignored; (2) selecting relevant covariates from high-dimensional spaces; and (3) accounting for uncertainty in estimated latent positions. Simulations show predictive performance remains stable as covariate sparsity grows, while naive approaches degrade. We illustrate how the method can support efficient study design using household social networks from 75 Indian villages, where an emulated pilot study screens a large covariate battery and substantially reduces required subsequent data collection without sacrificing network predictive accuracy.
\end{abstract}


\newpage

\section{Introduction}
\label{s:intro}
Network data arise throughout the social sciences and infectious disease epidemiology, where the scientific goal is often to relate patterns of interaction to node-level attributes. In addition to the complex dependencies inherent in graph data, two practical features complicate analysis: edge sparsity and high-dimensional node covariates. First, networks can be sparse with many isolated or low-degree nodes (e.g., contact and sexual networks), so individual ties carry disproportionate information and a substantial fraction of nodes contribute little edge information \citep{Weiss2020EgocentricStudy}. Second, investigators frequently collect a large battery of nodal measurements, only a subset of which are plausibly related to the social or epidemiological mechanisms that shape the network. In such settings, methods that incorporate covariates indiscriminately can dilute signal, increase computational burden, and hinder interpretability.

Latent space models provide a flexible and interpretable framework for representing network structure by positing that nodes occupy positions in a low-dimensional latent space, and that the propensity for tie formation is driven by similarity in those latent positions \citep{Hoff2002LatentAnalysis,Hoff2005BilinearData}. Conditional on latent positions, ties are modeled as independent, enabling likelihood-based estimation and inference while capturing prominent network features such as homophily, transitivity, and community structure through geometry \citep{Hoff2002LatentAnalysis,Smith2019TheData,Boguna2021NetworkGeometry}. A broad literature studies estimation and inference for latent position and related latent-structure random graph models \citep{Athreya2018StatisticalSurvey,Athreya2021OnGraphs,Hoff2002LatentAnalysis,Hoff2005BilinearData}. These models are increasingly used as a first-stage representation for downstream tasks, making reliable estimation of latent positions important for subsequent inference \citep{Athreya2021OnGraphs}.

In many empirical studies, rich node covariates are observed alongside the network. A common strategy is to incorporate covariates through dyadic similarity or edge-level covariates \citep{Hoff2002LatentAnalysis,Hoff2005BilinearData,Ma2020UniversalCovariates,Krivitsky2009RepresentingModels,Ma2022High-dimensionalLoss}, but this typically requires selecting a similarity metric and can be less informative for isolated nodes because edge-based likelihood contributions vanish. Approaches that fit latent space models with edge covariates can improve scalability and predictive performance, but are not designed for node covariate selection among high-dimensional node attributes and do not directly address the fact that the latent positions used as predictors are themselves estimated \citep{Ma2020UniversalCovariates}. More recently, \citet{Zhang2022JointVariables} introduced a joint latent space framework in which both the network and high-dimensional node variables are generated conditionally on shared latent positions. This joint formulation is particularly appealing for sparse networks because it leverages node variables even when nodes are isolated, thereby allowing all nodes to contribute to estimation.

Despite this progress, two gaps remain for practice. First, joint modeling alone does not resolve the common scenario in which $q$ node variables are collected, but only a small subset is meaningfully associated with the network-generating latent structure. In high dimensions, retaining many irrelevant covariates can worsen interpretability and may degrade predictive performance for the node covariate component. Second, any selection procedure that treats estimated latent positions $\widehat Z$ as fixed predictors risks overstating evidence: latent positions are estimated from the same data and may be noisy in sparse networks, creating an error-in-variables problem for covariate models built on $\widehat Z$.

We address these challenges by introducing a covariate selection procedure for joint latent space models that combines group-sparse regularization with a correction for latent position estimation error. Specifically, we treat each node covariate as a group corresponding to the $k$ coefficients linking that covariate to the $k$-dimensional latent space, and we apply a group lasso penalty for selection in high-dimensional logistic regression \citep{Meier2008TheRegression,Negahban2012ARegularizers}. To mitigate the impact of using estimated latent positions, we augment the group lasso objective with a stabilizing ridge-type term motivated by measurement-error-aware regularization for high-dimensional generalized linear models \citep{Srensen2018CovariateError,Srensen2015MeasurementCorrection}. Our overall procedure is a two-stage workflow: (i) joint latent space estimation using the full set of covariates following \citet{Zhang2022JointVariables}, and (ii) group-lasso screening (with measurement-error correction) followed by refitting the joint model on the selected covariate subset.

This paper makes three contributions. First, we propose a practically motivated covariate selection strategy for joint latent space modeling that is tailored to sparse networks and isolates. Second, we adapt measurement-error-aware penalization to the setting where latent positions are estimated and then used as predictors for high-dimensional node attributes \citep{Srensen2018CovariateError,Srensen2015MeasurementCorrection}. Third, we provide prediction error guarantees for the covariate component, both in an oracle setting where latent positions are treated as observed and in a setting where latent positions are estimated with bounded error; under uniform control across $q$ covariates and $n$ nodes, we obtain rates of order $O(\log q/n)$ under standard high-dimensional conditions \citep{Negahban2012ARegularizers,Meier2008TheRegression}. Finally, we illustrate an applied use case motivated by study design: an emulated pilot study that identifies a parsimonious covariate battery prior to full network data collection, demonstrated using social networks from 75 Indian villages \citep{Banerjee2013TheMicrofinance}.

\section{Model}
\label{s:model}
\subsection{Model Set Up and Notation}
Let $\mathcal{G} = (\mathcal{V},\mathcal{E})$ be an undirected network of $n$ nodes, with a symmetric adjacency matrix $A \in \mathbb{R}^{n \times n}$ such that $A_{ii', i \neq i'} = 1$ if the edge exists and 0 otherwise. For each node $i$, we observe $q$ binary covariates $Y \in \{0,1\}^{n \times q}$. We adopt a joint latent space model for the network data and node covariate data where each node $i$ is associated with an unobserved latent position $Z_i \in \mathbb{R}^k$ in a $k$-dimensional latent space. We assume both network edges and node covariates are generated conditional on these latent positions, creating a natural dependency structure between network topology and nodal attributes.

We assume that edges between nodes are independent conditional on their latent positions \citep{Ma2020UniversalCovariates, Zhang2022JointVariables, Hoff2002LatentAnalysis} such that
\[
    A_{ii'} = A_{i'i} \sim \text{Bernoulli}(P^A_{ii'})\:, \quad \text{logit} P^A_{ii'} = \theta^A_{ii'} = \alpha_i + \alpha_{i'}+ Z_i^\top Z_{i'}\;,
\]
where $P^A_{ii'}$ is the probability of an edge between nodes $i$ and $i'$ and $\alpha = (\alpha_1, \dots, \alpha_n)$ is a node-specific `sociability' factor that captures the propensity of nodes to form edges independent of their latent position. This formulation allows homophily, i.e., nodes with similar latent positions are more likely to connect, and degree heterogeneity.

For binary covariates, we assume that, conditional on $Z$, each covariate $Y_j$ for $j \in \{1, \dots, q\}$ follows an independent logistic regression
\[Y_{ij} \sim \text{Bernoulli}(P^Y_{ij})\;, \quad \text{logit} P^Y_{ij} = \theta^Y_{ij} = \gamma_j + Z_i\beta_{\cdot j}\;,\]
where $P^Y_{ij}$ is the probability that covariate $j$ for node $i$ is equal to 1, $\gamma \in \mathbb{R}^q$ is the vector of intercept terms, and $\beta \in \mathbb{R}^{k \times q}$ are the regression coefficients linking latent positions to node covariates.

Crucially, we assume that $\beta$ exhibits group sparsity: many columns of $\beta$ are zero, indicating that the corresponding covariate is not associated with the network-generating latent space. We define the active set  of $\beta$ as $S = \{j \in \{1, \dots, q\}: \|\beta_{\cdot j}\|_2 > 0\}$ with cardinality $s = |S|$. This assumption reflects that, while many covariates may be collected, only a subset of covariates are truly relevant to the network structure.

\subsection{Estimation}
Our estimation strategy proceeds in two stages: (1) initial joint estimation of all parameters using the full data, and (2) covariate selection followed by refined estimation using only selected covariates.

\vspace{0.5cm}

\textbf{Initial Estimation}

For identifiability, we center the latent variables by subtracting the row means and assume that $\frac{1}{n} Z^\top Z$ is diagonal but not identity. Thus, $Z$ is identifiable up to orthogonal transformations.

The joint estimation objective combines the negative log-likelihoods for both network and covariate data such that $L(Z, \alpha, \gamma, \beta) = L_A + \Lambda L_Y\;,$ where: 
\[L_A = -\log P(A|Z, \alpha) = \sum_{1\leq i < i' \leq n} \{A_{ii'} \theta^A_{ii'} - \log(1 + \exp(\theta^A_{ii'}))\}\;,\]
\[L_Y= -\log P(Y|\gamma, \beta) = \sum_{1\leq i \leq n, 1 \leq j \leq q} \{Y_{ij} \theta^Y_{ij} - \log(1 + \exp(\theta^Y_{ij}))\}\;.\]
The weight parameter $\Lambda \geq 0$ balances the information from network structure and covariates and is a hyperparameter. In practice, we use a small value, 0.1, to allow the network information greater importance during estimation. Projected gradient descent can be performed with ADAM \citep{Kingma2015Adam:Optimization} or Adagrad\citep{Duchi2011AdaptiveOptimization}.

Algorithm \ref{alg:stage1} presents a summarized algorithm. We use uninformative initializations: $Z^0 \sim N(0,1)$, $\alpha^0 \sim \text{Uniform}(-1,1)$, and $\beta^0$ and $\gamma^0$ as the output from a logistic regression using $Y$ and $Z^0$. Additionally, we require several hyperparameters. There are learning rates specific to $Z$ and $\alpha$ at each projected gradient descent step. As recommended in \cite{Zhang2022JointVariables}, we let $\eta_Z = \eta_t/ \|Z^{t-1}\|^2_F$ and $\eta_\alpha = \eta_t/(2n)$, depending on the step-size specific learning rate $\eta_t$, which is allowed to decay exponentially using cosine annealing. Given the exponential decay, we recommend a relatively large initial $\eta_t$. We allowed the algorithm to stop early if the mean log loss for $A$ and $Y$ had changed no more than $10^{-6}$ for 500 iterations.

\vspace{0.5cm}

\textbf{Covariate Selection}

Given the initial estimates $\hat{Z}$, $\hat{\alpha}$, $\hat{\beta}$, and $\hat{\gamma}$, each column $\beta_{\cdot j}$ is estimated as a single group of size $k$ via group lasso which accounts for estimation error in $\hat{Z}$. Given that each $Y_j$ is assumed to be independent conditional on $Z$, this becomes $q$ independent logistic regressions. Treating each column as a single group ensures that the estimates for each regression coefficient are penalized, inducing sparsity at the covariate level. 

For fixed and known latent positions, the group lasso estimator for each of the $q$ covariates would be 
\[\widehat{\beta}_{\cdot j} = \arg\min_{\beta_{\cdot j}} \left\{L_{Y_j}(\beta) + \lambda \sqrt{k} \|\beta_{\cdot j}\|_2 \right\}\;,
\]
However, because $Z$ must be estimated, we incorporate a correction for measurement error: 
\[
\widehat{\beta}_{\cdot j} = \arg\min_{\beta_{\cdot j}} \left\{L_{Y_j}(\beta) + \lambda  \sqrt{k} \|\beta_{\cdot j}\|_2 + \delta \sqrt{k} \|\beta_{\cdot j}\|_2^2 \right\}\;,
\]
This additional ridge-type penalty stabilizes the selection procedure in the presence of measurement error, as shown by \citet{Srensen2018CovariateError}. This selection procedure identifies the active set $\hat{S} = \{j \in \{1, \dots, q\} : \|\widehat{\beta}_{\cdot j}\| > \tau\}$ for an appropriate threshold $\tau$. Finally, we re-estimate all parameters using only the covariates in $\hat{S}$, yielding refined estimates that leverage the increased signal-to-noise ratio from covariate selection. 

A summarized version of the algorithm used can be found in Algorithm \ref{alg:stage2}. In addition to the hyperparameters used in Algorithm \ref{alg:stage1}, we also specify $\tau$ as a number very close to 0, a grid of candidate lasso penalties, $\lambda$, and a grid of candidate $\delta$ penalties. We let candidate $\lambda$ values be a grid of log-linear points including $\sqrt{k/n}$, and we let candidate $\delta$ values be a grid of points between 0 and 0.5 following the method of \citep{Srensen2018CovariateError}. The final value of $\lambda$ was chosen as the value which minimized the AIC but still retained some variables to balance the desire for parsimony without being overly aggressive in variable selection. Given that this is a nonconvex optimization, the result of the fitting algorithm was taken after the estimate had stabilized for 500 estimation steps (no change greater than $10^{-6}$) or as the result with the best per-parameter loss, calculated as $\frac{1}{n(n-1)}\text{log loss}(A) + \Lambda \frac{1}{nq} \text{log loss}(Y)$. The value of $\delta$ was chosen as that which, for the chosen $\lambda$, minimizes the log loss for $Y$.

\begin{algorithm}[htbp]
\caption{Joint latent space estimator}\label{alg:stage1}
\textbf{Input: }network $A\in\mathbb{R}^{n\times n}$; node covariates $Y\in\mathbb{R}^{n\times q}$; latent dimension $k$; initial $(Z^0,\alpha^0,\beta^0,\gamma^0)$; weighting parameter $\Lambda$; global step size $\eta$; maximum number of iterations $T$.
\BlankLine

\For{$t=1,\dots,T$}{
  Compute cosine annealed step size $\eta_t$ and set $\eta_z=\eta_t/\lVert Z^{t-1}\rVert_F^2$, $\eta_\alpha=\eta_t/(2n)$\;
  \emph{($Z,\alpha$ step)}: perform one optimization step on $(Z,\alpha)$ with step sizes $(\eta_z,\eta_\alpha)$\;
  $Z^t\leftarrow JZ^t$\;
  $\alpha^t \leftarrow \alpha^t - \bar{\alpha}^t\mathbf{1}_n$\;

  \For{$j = 1,\dots,q$}{
    Fit a logistic regression of $Y_{\cdot j}$ on $Z^t$ to obtain $(\gamma_j^t,\beta_{\cdot j}^t)$\;
  }
}

\textbf{Output: }estimates $(\widehat Z,\widehat \alpha,\widehat \beta,\widehat\gamma)$.
\end{algorithm}

\begin{algorithm}[htbp]
\caption{Lasso and ME--aware refinement to joint latent space estimators}\label{alg:stage2}
\textbf{Input: }Network adjacency matrix $A\in \mathbb{R}^{n\times n}$; node covariates $Y\in \mathbb{R}^{n\times q}$; latent dimension $k$; estimated $Z^{(1)}$ from Algorithm~\ref{alg:stage1}; weighting parameter $\Lambda$; global step size $\eta$; tolerance level $\tau$; iterations to perform lasso, $T$; grid of candidate lasso penalties $\lambda$; grid of candidate $\delta$ penalties $\Delta$.
\BlankLine

\ForEach{$\lambda_\ell \in \lambda$}{
  Compute $(\gamma^{\text{lasso}}(\lambda_\ell), \beta^{\text{lasso}}(\lambda_\ell))$ via column--wise group logistic lasso\;
  Calculate mean log-loss for $Y$, $\bar\ell_Y(\lambda_\ell)$, using $\gamma^{\text{lasso}}(\lambda_\ell), \beta^{\text{lasso}}(\lambda_\ell)$\;
}
$\lambda^* = \arg\min_{\lambda_\ell} \bar\ell_Y(\lambda_\ell)$ (given that some variables are still selected at $\lambda^*$)\;
\BlankLine

Compute $(\gamma^{\text{lasso}},\beta^{\text{lasso}})$ at $\lambda^*$\;
Remove columns $j$ with $\|\beta_{\cdot j}^{\text{lasso}}\|_2 < \tau$, yielding $Y_{\text{keep}}$, $\beta_{\text{keep}}$, $\gamma_{\text{keep}}$\;
\BlankLine

Using $Z^{(1)}$ and $(\gamma_{\text{keep}},\beta_{\text{keep}})$, run an iterative reweighting procedure over $\delta\in\Delta$ (see \cite{Srensen2018CovariateError})\;
Choose $(\gamma^{\text{GMUL}},\beta^{\text{GMUL}})$ that minimizes mean log-loss for $Y_{\text{keep}}$\;
\BlankLine

Starting from $(Z^{(1)},a^{(1)},\beta^{\text{GMUL}},\gamma^{\text{GMUL}})$ and using only $Y_{\text{keep}}$, rerun Algorithm~\ref{alg:stage1} for $T$ iterations, obtaining $(\widehat Z,\widehat a,\widehat \beta,\widehat\gamma)$\;

\textbf{Output: }estimates $(\widehat Z,\widehat a,\widehat \beta,\widehat\gamma)$.
\end{algorithm}

\section{Theoretical Results}
We establish the theoretical properties of our estimators under two scenarios: when latent positions are treated as observed without error (Section 3.1) and when accounting for estimation error (Section 3.2). We state the necessary assumptions and provide prediction error rates for each case.

\subsection{Perfectly observed $Z$}
First, we consider the idealized case where $Z$ is perfectly observed (known), which provides insight into the behavior of our selection procedure. Given that covariates are assumed to be independent, the problem reduces to group lasso for logistic regression. The group lasso estimator is $\widehat \beta_{\cdot j} \in \arg\min_{\beta_j \in \mathbb{R}^{k}} \Big\{ L_{Y_j}(\beta_{\cdot j}) + \tilde{\lambda}\sqrt{k}\, \|\beta_{\cdot j}\|_2 \Big\}$.

Given that $k$ is constant for all $\beta$, the degrees of freedom term can be absorbed into $\tilde{\lambda}$, yielding the familiar group lasso for logistic regression:
\[
\widehat \beta_{\cdot j} \in \arg\min_{\beta_j \in \mathbb{R}^{k}} \Big\{L_{Y_j}(\beta_{\cdot j}) + \lambda \|\beta_{\cdot j}\|_{2} \Big\}\;.
\]
We adopt the framework of \citet{Meier2008TheRegression} with modifications for our identifiability constraints.

\begin{assumption}[Bounded linear predictors]
\label{assumption1}
There exists $M_1 > 0$ such that $-M_1 < \Theta_{ii'}^A < M_1$ for $1 \leq i, i' \leq n$  and $M_2 > 0$ such that $-M_2 < \Theta_{j}^Y < M_2$ for $1 \leq j \leq q$. 
\end{assumption}

\begin{assumption}[Sparse Group Structure]
\label{assumption2}
Each column $\beta_{\cdot j}$ is treated as one group of size $k$ and the number of truly active covariates $s = |\{j : \|\beta_{\cdot j}\|_2 > 0\}|$ satisfies $s = o(n / \log q)$.
\end{assumption}

\begin{assumption}[Design Normalization]
\label{assumption3}
The empirical covariance $D = \tfrac{1}{n}Z^\top Z$ is diagonal (but not identity) with eigenvalues $0 < \nu_{\min}(D) \leq \nu_{\max}(D) < \infty$.
\end{assumption}

\begin{assumption}[Restricted Strong Convexity]
\label{assumption4}
Let the error between the estimated regression coefficients and the truth be $\widehat{\beta}_{\cdot j} - \beta_{\cdot j} = \Delta_j$. The notation $\langle \cdot,\cdot \rangle$ denotes an inner product. For each $j$, there exists $\kappa > 0$ such that for all $\Delta_j \in \mathbb{R}^k$,
\[
\frac{1}{n}\left(L_{Y_j}(\beta_{\cdot j} + \Delta_j; Z) - L_{Y_j}(\beta_{\cdot j}; Z) - \langle \nabla L_{Y_j}(\beta_{\cdot j}; Z), \Delta_j \rangle\right) \geq \frac{\kappa}{2}\|\Delta_j\|_2^2.
\]

\end{assumption}

\begin{assumption}[Score Concentration]
\label{assumption5}
For each $j$, with $\mu_{ij}(Z) = \text{logit}^{-1}(\gamma_j + Z_i^\top \beta_{\cdot j})$, we have
\[
\left\|\frac{1}{n}Z^\top\{Y_{\cdot j} - \mu_j(Z)\}\right\|_2 \leq C_0\sqrt{\frac{\log q}{n}}
\]
with high probability.
\end{assumption}

Assumption \ref{assumption1} is the same assumption used in \citet{Zhang2022JointVariables}; assumption \ref{assumption3} allows for greater identifiability of $Z$. Assumption \ref{assumption4} is a standard condition for high-dimensional M-estimators, ensuring that the loss function is strongly convex. The logistic loss is twice differentiable; thus, it is equivalent to a lower bound on the eigenvalues of the Hessian $\nabla^2L_{Y_j}(\beta_{\cdot j}, Z)$. It is well established by \citet{Negahban2012ARegularizers} and \citet{Meier2008TheRegression}. Assumption \ref{assumption5} follows from standard sub-exponential concentration \citep{Meier2008TheRegression}.

Define the prediction error $d^2(\widehat \eta_j, \eta_j) = \frac{1}{n}\|Z^\top(\widehat \beta_{\cdot j} - \beta_{\cdot j})\|_F^2$, $\eta_j = Z\beta_{\cdot j} $.

\begin{theorem}[Prediction Error Rate, Perfectly Observed $Z$]
\label{theorem1}
Under Assumptions \ref{assumption1}\textendash\ref{assumption5}: 

(a) For any fixed $j \in \{1,\dots,q\}$, with $\lambda_j = C\sqrt{k/n}$ for some constant $C$:
\[d^2(\widehat \eta_{\cdot j}, \eta_{\cdot j}) = O_p\Big(\frac{k}{n}\Big).\]

(b) For uniform control over all $j \in \{1,\dots,q\}$, with $\lambda = C\sqrt{\log q/n}$ for some constant $C$:
\[\max_{j=1,\dots,q} d^2(\widehat \eta_{\cdot j}, \eta_{\cdot j}) = O_p\Big(\frac{\log q}{n}\Big).\]
\end{theorem}

\citet{Meier2008TheRegression} states that $d^2(\hat{\eta}, \eta) = O_p(\frac{\log G}{n})$, where $G$ is the number of groups. For each separate optimization, we have a fixed group size of 1, so we recover the parametric convergence rate of $O_p(1/n)$. However, uniform control across $q$ columns requires the additional $\log q$ factor in the rate. The proof can be found in the appendix.

\subsection{Accounting for $Z$ with Measurement Error}
In practice, latent positions must be estimated, introducing additive measurement error: $\hat{Z} = Z + U$, where $Z$ are the oracle latent positions and $U$ is some unknown estimation error \citep{Srensen2018CovariateError, Srensen2015MeasurementCorrection}. The estimation-aware group lasso estimator for a single group is:
\begin{equation}
\widehat \beta_{\cdot j} \in \arg\min_{\beta_{\cdot j} \in \mathbb{R}^{k}} \Big\{ L_{Y_j}(\beta_{\cdot j}) +\ \lambda\,\|\beta_{\cdot j}\|_{2}\ +\ \frac{\delta}{2}\,\|\beta_{\cdot j}\|_{2}^2 \Big\} \; .
\end{equation}
Setting $\delta=0$ recovers the measurement-error-naive group lasso from the simplified case.

\begin{assumption}[Bounded measurement error]
\label{assumption6}
    Assume that the estimation error from the joint estimation stage satisfies $\|U\|_{\infty} \leq \delta_\infty$.
\end{assumption}

Note that \citet{Zhang2022JointVariables} states that $\frac{1}{n}\|\widehat Z - Z\|_F = O_p(n^{1/2})$, which allows us to assume $\delta_{\infty} = O(n^{-1/2})$. Thus, the assumption that the measurement error is bounded is not unreasonable. No other assumptions are made about the distribution of $U$.

\begin{theorem}[Prediction Error Rate, Estimated $\hat{Z}$]

Under Assumptions \ref{assumption1}\textendash\ref{assumption6}:

(a) For any fixed $j \in \{1,\dots,q\}$, with tuning parameters $\lambda_j = C_1\sqrt{k/n}$ and $\delta = C_2 \delta_{\infty}$ where $C_1, C_2$ are positive constants:
\[\frac{1}{n}\|\hat{Z}(\hat{\beta}_{\cdot j} - \beta_{\cdot j})\|_2^2 = O_p\Big(\frac{1}{n} + \delta^2_\infty\Big).\]

(b) For uniform control over all $j \in \{1,\dots,q\}$, with tuning parameters $\lambda = C_1\sqrt{\log q/n}$ and $\delta = C_2 \delta_{\infty}$:
\[\max_{j=1,\dots,q} \frac{1}{n}\|\hat{Z}(\hat{\beta}_{\cdot j} - \beta_{\cdot j})\|_2^2 = O_p\Big(\frac{\log q}{n} + \delta^2_\infty\Big).\]

When $\delta_{\infty} = O(n^{-1/2})$ from the first stage estimation, both rates reduce to their respective statistical error rates.

\end{theorem}

\textbf{Corollary 1:} When $\delta = O(n^{-1/2})$, the measurement error term is dominated by the statistical error, and we achieve the same rate as when $Z$ is perfectly observed. 

\section{Simulations}
\label{s:inf}
We evaluate our method through extensive simulations comparing three approaches: (1) network-only estimation ignoring covariates, (2) joint estimation without selection \citep{Zhang2022JointVariables}, and (3) our proposed method with covariate selection.

\subsection{Simulation Design}
We generated simulated data for $n = 200$ nodes, $q = 25$ node covariates, and $k = 2$ latent dimensions. To generate the networks and node covariate values, nodes were first divided into $k$ subsets with centers $\mu_k \sim \text{Uniform}(-1, 1)$. Latent positions $Z \in \mathbb{R}^{n \times k}$ were then generated such that, for node $i$ in subset $j$, $Z_i \sim N(\mu_j, 1)$. This simulates some amount of homophily in the network. Z was transformed to be centered and standardized with $Z = JZ$, $\| ZZ^\top\|_F = n$, with diagonal, non-identity empirical covariance. $\beta$ values were generated as $\beta \sim N(1, 0.1)$ for active components and set to $0$ otherwise. Finally, $A$ and $Y$ were generated according to the equations above. $A$ was forced to be symmetric by setting the lower triangle of the matrix equal to the upper triangle. Networks were generated with two levels of density by manipulating the node sociality parameters $\alpha = \big(\alpha_1, \dots, \alpha_n \big)$. A distribution of $\alpha_i \sim \text{Uniform}(-1, -0.5)$ generated networks with average density 0.37 (referred to as the less sparse network scenario below) and $\alpha_i \sim \text{Uniform}(-2, -1)$ generated networks with average density 0.013 (referred to as the sparse network scenario below). We vary the covariate sparsity level across simulations, with each configuration repeated 30 times. 

\subsection{Simulation Results}
Results in Figure \ref{boxplot_sbs} demonstrate the result of modeling 30 independent simulated networks at different levels of covariate sparsity with four different methods: using network data alone, incorporating network data without adjusting for sparsity \citep{Zhang2022JointVariables}, adjusting for sparsity (lasso), and adjusting for sparsity as well as measurement error (meLasso). We present results using the area under the curve (AUC) metric for network information $A$ and $Y$, averaged over all nodes in the network, to determine model performance. Importantly, results are presented as the average AUC for all covariates included in the model; if the sparsity-aware methods dropped a covariate, it was not included in the final average AUC. Results are presented in Table \ref{tab_aucY}.

Simulation results are presented in Figure \ref{boxplot_sbs} and Table \ref{tab_aucY}. For both less-sparse networks and sparse networks, model performance in predicting the network is robust to sparsity in the node covariates for all methods, as can be seen in panels A and B in Figure \ref{boxplot_sbs}. Given that the network and covariate information balancing parameter we used was small, $\Lambda = 0.1$, this was to be expected; the estimation procedure was weighted much more strongly towards the network information, and therefore information from the covariates had a limited impact on the estimation of the latent space positions informing the network edges. Additionally, when comparing the results of the sparse and less sparse networks, the model was better able to capture network information when the underlying network was denser, resulting in AUC for the network information being closer to 0.65 regardless of the modeling method, while it was closer to 0.57 for the sparse networks.

When modeling node covariates, model performance is sensitive to both the sparsity of the network and the amount of noise in the node covariates. Panel D in Figure \ref{boxplot_sbs} demonstrates that, for the less sparse networks, the latent positions are modeled sufficiently well so that the primary determinant of the average AUC is the number of covariates included in the modeling approach that are actually related to the latent space. When there are 0 noise covariates, all of the models can predict the node covariates well on average. However, as the number of noise covariates increases, methods that do not consider the sparsity in the node covariates see degraded performance because the AUC is still averaging over covariates that are not related to the latent space. However, the lasso and meLasso approaches were able to remove approximately 86\% of the noise covariates (Table \ref{tab_varselect}), resulting in a much higher average AUC. 

If the network is sparse, the benefit of including node covariate information in the joint estimation procedure becomes more apparent (Panel C in Figure \ref{boxplot_sbs}). Even when all covariates are associated with the latent space, the modeling approach that uses the network data alone performs substantially worse than the other methods that include covariate information. As the number of noise covariates increases, the benefit of including covariate information during estimation decreases unless we account for potential sparsity. Once the majority of the node covariates are noise, we can see that the lasso and meLasso approaches are more robust than the method that does not adjust for sparsity, though the results are highly variable. This is due in large part because the lasso procedure struggled to identify noise covariates when the network information was sparse; when 20 of the 25 node covariates are noise, lasso was only able to drop 12\% of the noise covariates on average and meLasso was only able to drop 56\% on average (Table \ref{tab_varselect}).

\begin{figure}[htbp]
    \centerline{%
    \includegraphics[width = 1.05\textwidth]{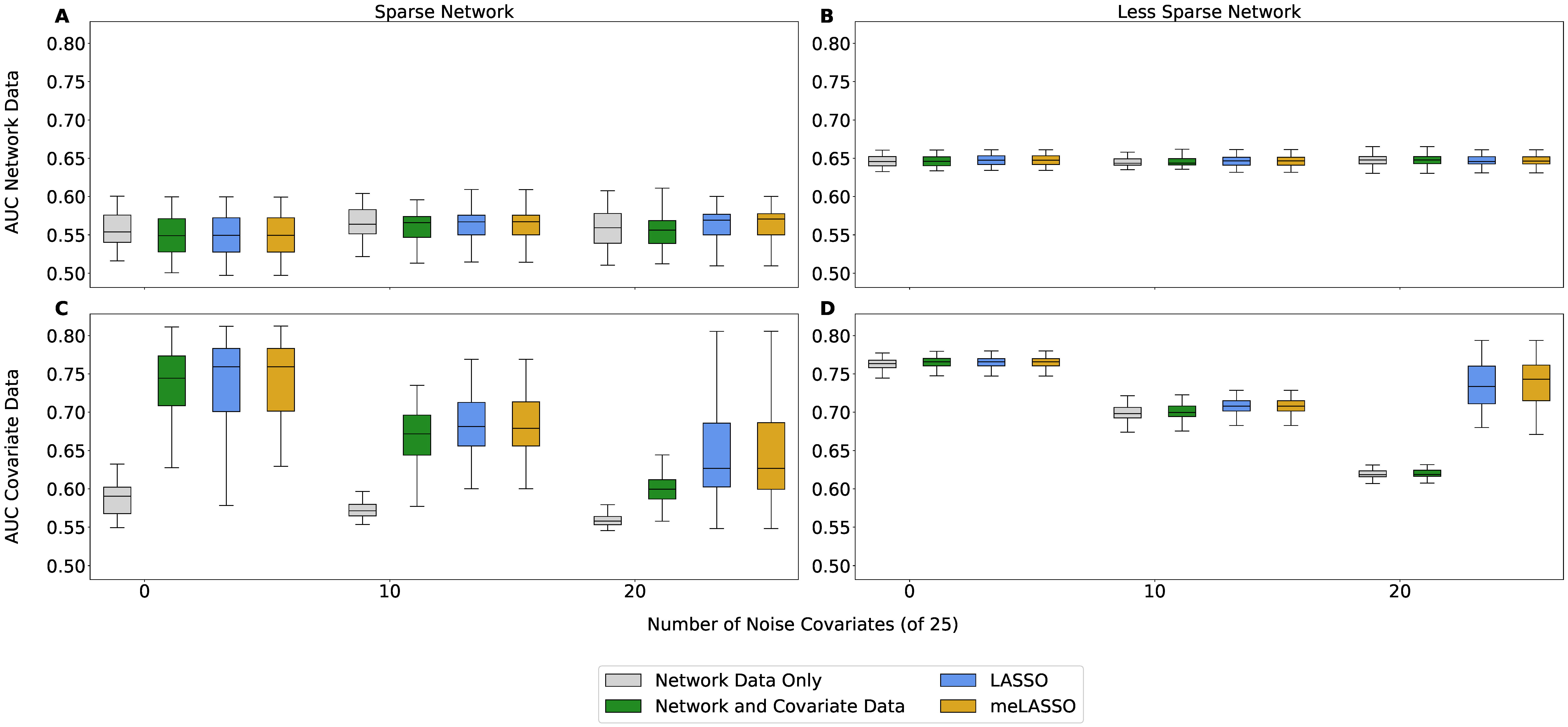}}
    \caption{\textbf{Model Performance in Simulations.} Boxplots represent the results from 30 independently simulated networks of $n = 200$ nodes. Sparse networks have an average density of 0.013 (average degree = 2.6); less sparse networks have an average density of 0.37 (average degree = 74). Panels A and B show the AUC for network data in the sparse and less sparse networks, respectively. Panels C and D show the AUC for node covariate data for the sparse and less sparse networks, respectively.}
    \label{boxplot_sbs}
\end{figure}

\begin{table}[htbp]
\caption{Performance in predicting $q = 25$ node covariates for each method under varying numbers of noise covariates for $n = 200$ nodes.}
\centering
\begin{tabular}{lccc}
\hline
\textbf{Method} & \textbf{0 noise covariates} & \textbf{10 noise covariates} & \textbf{20 noise covariates} \\
& Mean (SD) & Mean (SD) & Mean (SD) \\
\hline
& \multicolumn{3}{c}{\textbf{Less sparse network (density = 0.37)}}\\
\hline
Network data only     & 0.76 (0.01) & 0.70 (0.01) & 0.62 (0.01) \\
Network and covariate data  & 0.77 (0.01) & 0.70 (0.01) & 0.62 (0.01) \\
lasso     & 0.77 (0.01) & 0.71 (0.02) & 0.73 (0.04) \\
meLasso  & 0.77 (0.01) & 0.71 (0.01) & 0.74 (0.04) \\
& \multicolumn{3}{c}{\textbf{Sparse network (density = 0.013)}}\\
\hline
Network data only    & 0.59 (0.02) & 0.57 (0.01) & 0.56 (0.01) \\
Network and covariate data    & 0.73 (0.06) & 0.67 (0.04) & 0.60 (0.02) \\
lasso     & 0.74 (0.06) & 0.68 (0.05) & 0.69 (0.14) \\
meLasso  & 0.74 (0.06) & 0.68 (0.05) & 0.69 (0.14) \\
\hline
\bottomrule
\end{tabular}
\label{tab_aucY}
\end{table}

\begin{table}[htbp]
\caption{Variable selection with lasso and measurement-aware lasso.
Each cell reports the average proportion of noise covariates removed (true negatives, TN) and the proportion of true covariates remaining (true positives, TP) across 30 independent realizations.}
\centering
\begin{tabular}{lccc}
\hline
\textbf{Method} & \textbf{0 noise covariates} & \textbf{10 noise covariates} & \textbf{20 noise covariates} \\
& TN, TP & TN, TP &  TN, TP\\
\hline
& \multicolumn{3}{c}{\textbf{Less sparse network (density = 0.37)}}\\
\hline
lasso     & NA, 1 & 0.12, 1 & 0.86, 1 \\
meLasso  & NA, 1 & 0.12, 1 & 0.87, 1\\
& \multicolumn{3}{c}{\textbf{Sparse network (density = 0.013)}}\\
\hline
lasso     & NA, 1 & 0.03, 1 & 0.12, 0.97 \\
meLasso  & NA, 1 & 0.12, 0.97 & 0.56, 0.64\\
\hline
\bottomrule
\end{tabular}
\label{tab_varselect}
\end{table}

\section{Application to Real Data}
We demonstrate how this method can guide investigators planning large network studies. Given that ascertaining networks is time and labor-intensive, reducing data collection lessens participant burden and allows investigators to focus on better capture of more meaningful information.

We apply our model to data from a 2006 baseline survey of social networks in 75 villages in rural southern India conducted by \citet{Banerjee2013TheMicrofinance, dataset2013TheData}. Each household was surveyed on amenities, such as access to electricity, A more detailed survey was given to approximately half of households, in which individuals were asked information about their social network, defined as people who visit their home or whose home they visit, family ties, people they give or receive advice from, give or receive material support from (money, goods), and with whom they pray. This information was used to generate undirected, household-level social networks for each village. The village networks are effectively independent due to the distance between each village. The average network density was 0.05.

The dataset contains numeric and categorical data. We binarized the node (household) covariates such that the binary variable is relatively balanced, if possible, and variables associated with economic status, such as access type of roof material, generally represent higher or lower socioeconomic status. This is described in Table 2. We augmented the 8 real variables from the study by creating 24 random variables, which were Bernoulli(0.5). We then emulated a pilot study by randomly selecting 10 villages. Measurement error-aware lasso identified covariates associated with the network latent space. Covariates dropped by at least $70\%$ of pilot villages were excluded from the full analysis unless they were rare (prevalence $< 10\%$ in any pilot village). This dropped one of the original covariates and 21 of the 24 of the noise covariates, reducing the final data collection required by 69\%. We then analyzed the data from the remaining 65 villages using only the covariates selected in the pilot study and compared the results to using all 32 covariates. The variable original to the study that was dropped in this process was the household religion, likely because more than half of the villages were entirely Hindu. Thus, this variable was unlikely to be informative given the lack of variability.

\begin{table}[ht]
\centering
\caption{Binarization of household covariates.}
\label{tab:binarization}
\begin{tabular}{lcp{5cm}l}
\hline
\toprule
\textbf{Variable} & \textbf{Value} & \textbf{Definition} & \textbf{Average Prevalence}\\
\midrule
\multirow{2}{6em}{Religion}
  & 1 & Household religion is Hinduism & 96.0\%\\
  & 0 & Any other religion & 4\%\\
\hline
\multirow{2}{6em}{Caste}
  & 1 & Caste or subcaste is 'Schedule Caste' or 'Schedule Tribe' & 20.5\%\\
  & 0 & All other castes & 79.5\%\\
  \hline
\multirow{2}{6em}{Number of rooms}
  & 1 & Number of rooms $>2$ & 34.8\%\\
  & 0 & Two rooms or fewer & 65.2\%\\
  \hline
\multirow{2}{6em}{Number of beds}
  & 1 & Number of beds $>0$ & 48.8\%\\
  & 0 & No beds & 51.2\%\\
  \hline
\multirow{2}{6em}{House has electricity}
  & 1 & Yes, private & 61.3\%\\
  & 0 & No, or government-provided & 38.7\%\\
  \hline
\multirow{2}{6em}{Latrine}
  & 1 & Latrine present & 26.0\%\\
  & 0 & None & 74.0\%\\
  \hline
\multirow{2}{6em}{Roof type}
  & 1 & Either thatch or tile & 33.9\%\\
  & 0 & Stone, sheet metal, or reinforced cement concrete & 66.1\%\\
  \hline
\multirow{2}{6em}{Community leader}
  & 1 & Yes & 12.6\%\\
  & 0 & No & 87.4\%\\
\hline
\bottomrule
\end{tabular}
\label{tab_aucY}
\end{table}

Network prediction AUC was unchanged (mean difference = 0.00, standard deviation (SD) = 0.002). Comparing the mean AUC for predicting all 32 covariates to the mean AUC predicting the selected 6 covariates, the AUC for the restricted variable set was 0.09 higher on average (SD: 0.02). When AUC is calculated only among the restricted variable sets, the difference disappears, which can be seen in Figure \ref{real_dat_fig}. For the majority of villages, the estimates obtained with meLasso and without on the restricted variable set are equivalent. Of the three cases where the estimates differed by more than 0.01, it was in favor of meLasso. This indicates that utilizing measurement error-aware lasso to select a subset of variables is not detrimental to estimation and does screen out covariates that are not associated with the latent space, and therefore cannot be predicted with the available information.

\begin{figure}
    \centerline{%
    \includegraphics[width = \textwidth]{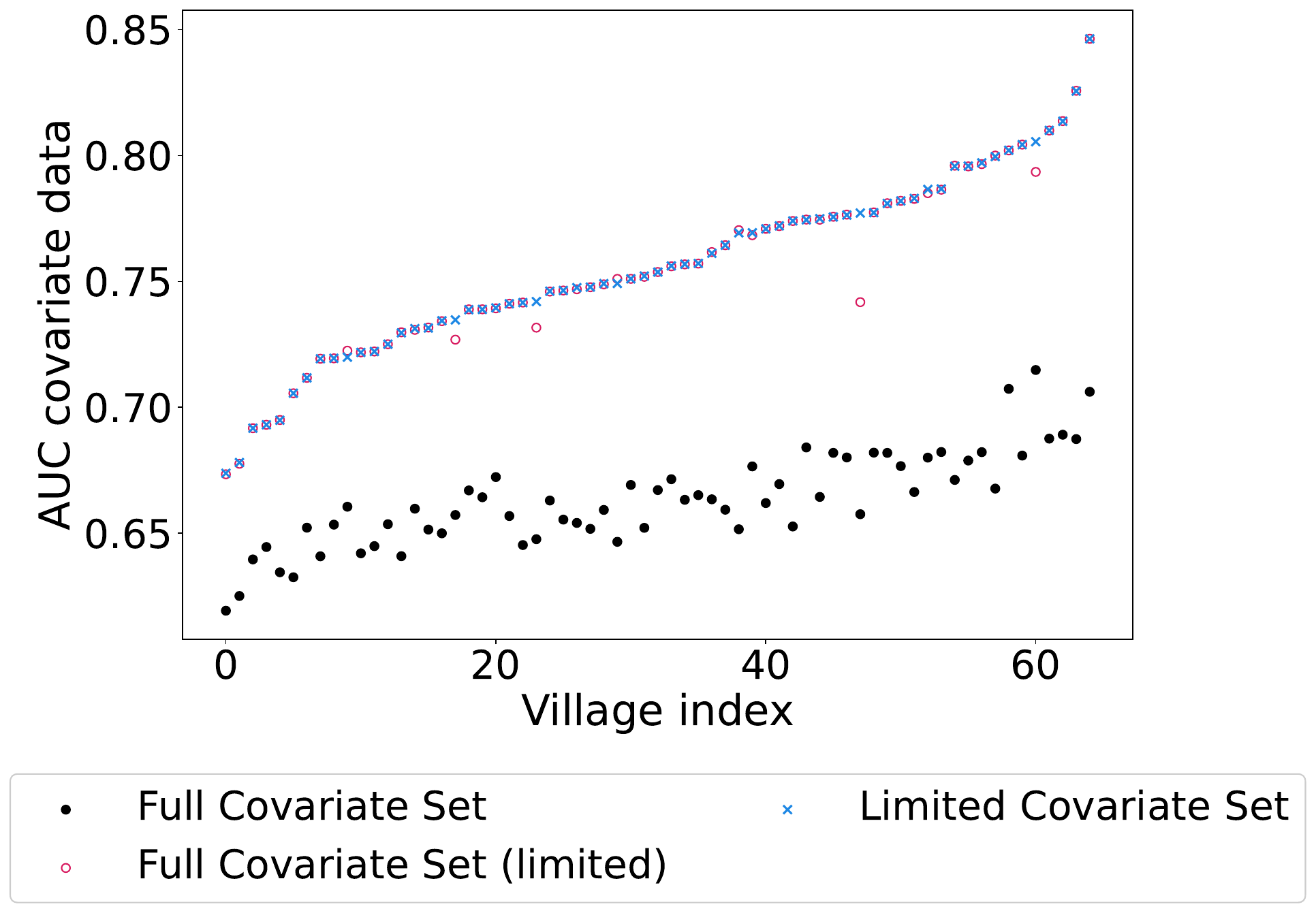}}
    \caption{\textbf{Model performance predicting node covariates with and without an emulated pilot study.} The figure shows the average AUC for node covariates on 65 villages not selected to be in a pilot study. Markers in black show the average AUC over all 32 variables when estimated on the full covariate set. 'X' Markers in blue show the average AUC when modeling only the covariates identified as of interest in the emulated pilot study. Markers in red, described as "full covariate set (limited)", show the average AUC on variables retained by the pilot study but estimated using all covariates. Villages are ordered by increasing AUC for the limited covariate set.}
    \label{real_dat_fig}
\end{figure}

This data is publicly available via the Harvard Dataverse at \url{https://doi.org/10.7910/DVN/U3BIHX}. All the code used in this project is available at \url{https://github.com/onnela-lab/lsm-lasso}.

\section{Discussion}
\label{s:discuss}

We proposed a covariate selection framework for joint latent space modeling of network data with high-dimensional binary node attributes. We exploit the joint model structure of \citet{Zhang2022JointVariables} which allows node attributes to inform latent position estimation even when networks are sparse and many nodes are isolated while introducing a group-sparse selection step that identifies which node covariates are most closely associated with the network-generating latent space. Our approach uses group lasso for logistic regression \citep{Meier2008TheRegression} and incorporates a ridge-type stabilization term motivated by measurement-error-aware regularization \citep{Srensen2018CovariateError,Srensen2015MeasurementCorrection} to account for uncertainty in estimated latent positions.

Our theoretical results establish prediction error bounds for the covariate component, both when latent positions are treated as fixed or observed and when they are estimated with bounded error. Under uniform control across $q$ covariates, the resulting rate is of order $O(\log q/n)$ up to an additional term reflecting latent position estimation error, consistent with high-dimensional M-estimation theory for generalized linear models with decomposable regularizers \citep{Negahban2012ARegularizers,Meier2008TheRegression}. These results provide formal support for using the proposed screening-and-refit strategy to maintain predictive accuracy in high-dimensional settings. 

Across simulations, selection procedures that account for covariate sparsity were substantially more robust than approaches that included all node covariates, particularly as the proportion of noise covariates increased. The real-data analysis illustrates a complementary and practically motivated use case: an emulated pilot study that screens a large covariate battery before full network ascertainment. In settings where surveys are burdensome and network measurement is expensive, this can support efficient data collection by focusing measurement on variables most strongly linked to the network-generating latent structure, while preserving predictive performance for the network component \citep{Banerjee2013TheMicrofinance}. These results were strong despite selecting the pilot villages entirely at random, while a real pilot study would likely take much greater care to ensure that the villages selected were a representative sample of the population of interest. Perhaps most importantly, the simulation and real data examples show that there can be substantial benefit to using sparsity-aware estimation techniques when including node covariate information in latent space modeling, but also that there are very few consequences to doing so. In our simulations, true covariates were rarely dropped unless both the network and the node covariate information were very sparse, a situation in which there was the least amount of information with which to estimate the modeling parameters. Additionally, the real data results show both that the only true variable dropped likely contained little information about the network due to lack of variability, and that collecting the full data did not result in better estimation or prediction of the covariates selected by measurement-aware lasso.

Model performance depends on a small number of tuning and modeling choices. First, the latent dimension $k$ can be selected using a modest candidate set and choosing the value that optimizes out-of-sample predictive performance for the network component, following common practice in latent position models \citep{Hoff2002LatentAnalysis,Hoff2005BilinearData}. Second, the balance parameter $\Lambda$ controls the relative influence of the node attributes versus the network in the joint objective \citep{Zhang2022JointVariables}. In sparse networks, moderate values of $\Lambda$ can help stabilize latent position estimation by leveraging node covariates for low-degree and isolated nodes, but overly large values can shift the latent space toward explaining the covariate likelihood at the expense of network fit. Third, the selection penalty $\lambda$ governs sparsity: larger $\lambda$ yields fewer selected covariates, and we found it useful to constrain the search to solutions that retain at least one covariate to avoid degenerate ``all-zero'' fits in pilot-screening contexts. Finally, the measurement-error stabilization parameter $\delta$ trades off bias and variance: larger $\delta$ shrinks coefficients more aggressively and can reduce instability induced by noisy $\widehat Z$, consistent with prior findings in measurement-error-aware penalization \citep{Srensen2018CovariateError,Srensen2015MeasurementCorrection}. 

Several limitations warrant attention. First, we focused on binary node covariates with logistic models. Extending to continuous or mixed-type covariates is conceptually straightforward within generalized linear modeling, but in practice requires careful scaling and regularization calibration to ensure comparability of penalties across covariate types \citep{Zhang2022JointVariables}. Second, we assumed an undirected network. Directed networks would require modeling asymmetric tie formation and potentially distinct latent representations for sending and receiving, as well as modified degree-heterogeneity terms \citep{Hoff2002LatentAnalysis,Krivitsky2009RepresentingModels}. Third, our estimation objective is nonconvex, and while first-order methods such as ADAM and Adagrad are effective in practice \citep{Kingma2015Adam:Optimization}, the presence of local optima implies that initialization and optimization diagnostics remain important. Reporting sensitivity to multiple random restarts and assessing stability of the selected covariate set are recommended in applied analyses.

There are several promising extensions. One is to strengthen the model selection theory by characterizing conditions under which exact support recovery is achievable in the latent-position-as-regressor setting. Another direction is computational scaling. For larger networks, developing approximations or exploiting sparsity-aware optimization may substantially reduce runtime. Finally, it would be valuable to study robustness to model misspecification, including deviations from the assumed latent space geometry and additional structural features such as clustering and degree heterogeneity beyond what is captured by the current specification \citep{Athreya2018StatisticalSurvey,Athreya2021OnGraphs}.

The increasing availability of network data paired with rich node-level measurements motivates methods that are both statistically principled and practically useful for sparse networks. By combining joint latent space modeling with group-sparse covariate screening and a measurement-error-aware stabilization, our framework provides a tractable approach to identify and leverage node attributes that align with the network-generating latent structure, supporting both interpretability and efficient study design in network-based social and epidemiological research.


\section*{Declarations}

\subsection*{Availability of data and materials}

This data is publicly available via the Harvard Dataverse at \texttt{https://doi.org/10.7910/DVN/U3BIHX}. All the code used in this project is available at \texttt{https://github.com/onnela-lab/lsm-lasso}.

\subsection*{Funding}

This project was supported by the NIH grants R01AI138901 and T32AI007358. The funder had no role in the study design, collection, analysis, or interpretation of the data, writing the manuscript, or the decision to submit the paper for publication.

\subsection*{Authors' contributions}

E.G.C. and J.P.O. designed the research; E.G.C. performed the research. E.G.C. and Y.Z. wrote and edited the theoretical proofs. E.G.C., Y.Z., and J.P.O. wrote and edited the paper. J.P.O. supervised the research.

\subsection*{Acknowledgments}

We thank the members of the Onnela lab for their insights and feedback on this project.

\subsection*{Declaration of the use of generative AI and AI-assisted technologies}
During the preparation of this work, the authors used ChatGPT 5.0 to assist in LaTeX typesetting and fine-tuning of figure appearance. After using this tool, the authors reviewed and edited the content as necessary and take full responsibility for the content of the publication.

\section*{Supplementary material}
The supplementary information contains proofs of the theoretical results (Theoreoms 1 and 2).

\section*{References}

\newpage
\bibliographystyle{apalike} %
\bibliography{paper2_ref}%


\newpage

\appendix

\appendixone
\section*{Appendix 1}
\setcounter{theorem}{0}
\setcounter{assumption}{0}

\subsection{Reminder of Model Notation and Assumptions}
Let $\mathcal{G} = (\mathcal{V},\mathcal{E})$ be an undirected network of $n$ nodes, with a symmetric adjacency matrix $A \in \mathbb{R}^{n \times n}$ such that $A_{ii', i \neq i'} = 1$ if the edge exists and 0 otherwise. For each node $i$, we observe $q$ binary covariates $Y \in \{0,1\}^{n \times q}$ and each node $i$ is associated with an unobserved, low-dimensional latent position $Z_i \in \mathbb{R}^k$ for some fixed $k$. We center the latent variables by subtracting the row means and require that the empirical covariance of $Z$, $\frac{1}{n} Z^\top Z$ is diagonal but not identity.

We assume that edges between nodes are independent conditional on their latent positions such that
\[
    A_{ii'} = A_{i'i} \sim \text{Bernoulli}(P^A_{ii'})\:, \quad \text{logit} P^A_{ii'} = \theta^A_{ii'} = \alpha_i + \alpha_{i'}+ Z_i^\top Z_{i'}\;,
\]
where $P^A_{ii'}$ is the probability of an edge between nodes $i$ and $i'$ and $\alpha = (\alpha_1, \dots, \alpha_n)$ is a node-specific `sociability' factor that captures the propensity of nodes to form edges independent of their latent position.

For binary covariates, we assume that, conditional on $Z$, each covariate $Y_j$ for $j \in \{1, \dots, q\}$ follows an independent logistic regression
\[Y_{ij} \sim \text{Bernoulli}(P^Y_{ij})\;, \quad \text{logit} P^Y_{ij} = \theta^Y_{ij} = \gamma_j + Z_i\beta_{\cdot j}\;,\]
where $P^Y_{ij}$ is the probability that covariate $j$ for node $i$ is equal to 1, $\gamma \in \mathbb{R}^q$ is the vector of intercept terms and $\beta \in \mathbb{R}^{k \times q}$ are the regression coefficients linking latent positions to node covariates.

Crucially, we assume that $\beta$ exhibits group sparsity: many columns of $\beta$ are zero, indicating that the corresponding covariate is not associated with the network-generating latent space. We define the active set  of $\beta$ as $S = \{j \in \{1, \dots, q\}: \|\beta_{\cdot j}\|_2 > 0\}$ with cardinality $s = |S|$.

Our assumptions are as follows:

\begin{assumption}[Bounded linear predictors]
\label{assumption1}
There exists $M_1 > 0$ such that $-M_1 < \Theta_{ii'}^A < M_1$ for $1 \leq i, i' \leq n$  and $M_2 > 0$ such that $-M_2 < \Theta_{j}^Y < M_2$ for $1 \leq j \leq q$. 
\end{assumption}

\begin{assumption}[Sparse Group Structure]
\label{assumption2}
Each column $\beta_{\cdot j}$ is treated as one group of size $k$ and the number of truly active covariates $s = |\{j : \|\beta_{\cdot j}\|_2 > 0\}|$ satisfies $s = o(n / \log q)$.
\end{assumption}

\begin{assumption}[Design Normalization]
\label{assumption3}
The empirical covariance $D = \tfrac{1}{n}Z^\top Z$ is diagonal (but not identity) with eigenvalues $0 < \nu_{\min}(D) \leq \nu_{\max}(D) < \infty$.
\end{assumption}

\begin{assumption}[Restricted Strong Convexity]
\label{assumption4}
Let the error between the estimated regression coefficients and the truth be $\widehat{\beta}_{\cdot j} - \beta_{\cdot j} = \Delta_j$. The notation $\langle \cdot,\cdot \rangle$ denotes an inner product. For each $j$, there exists $\kappa > 0$ such that for all $\Delta_j \in \mathbb{R}^k$,
\[
\frac{1}{n}\left(L_{Y_j}(\beta_{\cdot j} + \Delta_j; Z) - L_{Y_j}(\beta_{\cdot j}; Z) - \langle \nabla L_{Y_j}(\beta_{\cdot j}; Z), \Delta_j \rangle\right) \geq \frac{\kappa}{2}\|\Delta_j\|_2^2.
\]

\end{assumption}

\begin{assumption}[Score Concentration]
\label{assumption5}
For each $j$, with $\mu_{ij}(Z) = \text{logit}^{-1}(\gamma_j + Z_i^\top \beta_{\cdot j})$, we have
\[
\left\|\frac{1}{n}Z^\top\{Y_{\cdot j} - \mu_j(Z)\}\right\|_2 \leq C_0\sqrt{\frac{\log q}{n}}
\]
with high probability.
\end{assumption}

\begin{assumption}[Bounded measurement error]
\label{assumption6}
    Assume that the estimation error from the joint estimation stage satisfies $\|U\|_{\infty} \leq \delta_\infty$.
\end{assumption}

Assumption \ref{assumption1} is the same assumption used in \citet{Zhang2022JointVariables}; assumption \ref{assumption3} allows for greater identifiability of $Z$. Assumption \ref{assumption4} is a standard condition for high-dimensional M-estimators, ensuring that the loss function is strongly convex. The logistic loss is twice differentiable; thus, it is equivalent to a lower bound on the eigenvalues of the Hessian $\nabla^2L_{Y_j}(\beta_{\cdot j}, Z)$. It is well established by \citet{Negahban2012ARegularizers} and \citet{Meier2008TheRegression}. Assumption \ref{assumption5} follows from standard sub-exponential concentration \citep{Meier2008TheRegression}. Finally, \citet{Zhang2022JointVariables} states that $\frac{1}{n}\|\widehat Z - Z\|_F = O_p(n^{1/2})$, which allows us to assume $\delta_{\infty} = O(n^{-1/2})$. Thus, assumption \ref{assumption6}, that the measurement error is bounded, is not unreasonable. No other assumptions are made about the distribution of $U$.

The joint estimation objective combines the negative log-likelihoods for both network and covariate data such that $L(Z, \alpha, \gamma, \beta) = L_A + \Lambda L_Y\;,$ where: 
\[L_A = -\log P(A|Z, \alpha) = \sum_{1\leq i < i' \leq n} \{A_{ii'} \Theta^A_{ii'} - \log(1 + \exp(\Theta^A_{ii'}))\}\;,\]
\[L_Y= -\log P(Y|\gamma, \beta) = \sum_{1\leq i \leq n, 1 \leq j \leq q} \{Y_{ij} \Theta^Y_{ij} - \log(1 + \exp(\Theta^Y_{ij}))\}\;.\]
The weight parameter $\Lambda \geq 0$ balances the information from network structure and covariates.

\subsection{Proof of Theorem \ref{theorem1}}
Define the prediction error $d^2(\widehat \eta_j, \eta_j) = \frac{1}{n}\|Z^\top(\widehat \beta_{\cdot j} - \beta_{\cdot j})\|_F^2$, $\eta_j = Z\beta_{\cdot j} $.

\begin{theorem}[Prediction Error Rate, Perfectly Observed $Z$]
\label{theorem1}
Under Assumptions \ref{assumption1}\textendash\ref{assumption5}: 

(a) For any fixed $j \in \{1,\dots,q\}$, with $\lambda_j = C\sqrt{k/n}$ for some constant $C$:
\[d^2(\widehat \eta_{\cdot j}, \eta_{\cdot j}) = O_p\Big(\frac{k}{n}\Big).\]

(b) For uniform control over all $j \in \{1,\dots,q\}$, with $\lambda = C\sqrt{\log q/n}$ for some constant $C$:
\[\max_{j=1,\dots,q} d^2(\widehat \eta_{\cdot j}, \eta_{\cdot j}) = O_p\Big(\frac{\log q}{n}\Big).\]
\end{theorem}

\begin{proof}
\label{proof1}
\textbf{Part (a): Individual regression rate.}

For a fixed $j$, by adapting the approach of \citet{Meier2008TheRegression} for a single regression with one group of size $k$, and using the score concentration from Assumption \ref{assumption5} with rate $\sqrt{k/n}$, the group lasso estimator with penalty $\lambda_j = C\sqrt{k/n}$ satisfies:
\[
d^2(\widehat \eta_j, \eta_j) = O_p\Big(\frac{k}{n}\Big).
\]

Because $D = n^{-1}Z^\top Z$ is diagonal with bounded eigenvalues, we have norm equivalence
\[
\nu_{\min}(D)\, \frac{1}{n}\|\tilde Z^\top(\widehat \beta_{\cdot j} - \beta_{\cdot j})\|_F^2 \leq d^2(\widehat \eta_j, \eta_j) \leq \nu_{\max}(D)\, \frac{1}{n}\|\tilde Z^\top(\widehat \beta_{\cdot j} - \beta_{\cdot j})\|_F^2 \; ,
\]
where $\tilde Z = ZD^{-1/2}$. Thus, the rate holds for $Z$ up to a constant depending on $D$.

\textbf{Part (b): Uniform control via union bound.}

To ensure the optimization error is controlled uniformly over all $q$ columns of $\beta$, we need the penalty to dominate the maximum score fluctuation across all $q$ regressions. By choosing $\lambda = C\sqrt{\log q/n}$ with $C$ sufficiently large (specifically, $C > C_0\sqrt{k/\log q}$), we obtain
\[
\max_{j=1,\dots,q} d^2(\widehat \eta_j, \eta_j) = O_p\Big(\frac{\log q}{n}\Big)
\]

via standard union bound arguments.
\end{proof}

\subsection{Proof of Theorem 2}
In practice, latent positions are unknown and must be estimated, introducing additive measurement error. The main difficulty of this theorem is the additional bias introduced by using $\hat{Z}$ instead of $Z$.

Let $\hat{Z} = Z + U$, where $Z$ are the oracle latent positions and $U$ is some unknown estimation error \citep{Srensen2018CovariateError, Srensen2015MeasurementCorrection}. The estimation-aware group lasso estimator for a single group is:
\begin{equation}
\widehat \beta_{\cdot j} \in \arg\min_{\beta_{\cdot j} \in \mathbb{R}^{k}} \Big\{ L_{Y_j}(\beta_{\cdot j}) +\ \lambda\,\|\beta_{\cdot j}\|_{2}\ +\ \frac{\delta}{2}\,\|\beta_{\cdot j}\|_{2}^2 \Big\} \; .
\end{equation}
Setting $\delta=0$ recovers the measurement-error-naive group lasso from the simplified case.

\begin{theorem}
\label{theorem2}
Under Assumptions \ref{assumption1}\textendash\ref{assumption6}:

(a) For any fixed $j \in \{1,\dots,q\}$, with tuning parameters $\lambda_j = C_1\sqrt{k/n}$ and $\delta = C_2 \delta_{\infty}$ where $C_1, C_2$ are positive constants:
\[\frac{1}{n}\|\hat{Z}(\hat{\beta}_{\cdot j} - \beta_{\cdot j})\|_2^2 = O_p\Big(\frac{1}{n} + \delta^2_\infty\Big).\]

(b) For uniform control over all $j \in \{1,\dots,q\}$, with tuning parameters $\lambda = C_1\sqrt{\log q/n}$ and $\delta = C_2 \delta_{\infty}$:
\[\max_{j=1,\dots,q} \frac{1}{n}\|\hat{Z}(\hat{\beta}_{\cdot j} - \beta_{\cdot j})\|_2^2 = O_p\Big(\frac{\log q}{n} + \delta^2_\infty\Big).\]

When $\delta_{\infty} = O(n^{-1/2})$ from the first stage estimation, both rates reduce to their respective statistical error rates.

\end{theorem}

\begin{proof}
\label{proof2}
We begin by proving part (a) and then prove part (b) with several additional steps.

Let $\Delta_j = \widehat{\beta}_{\cdot j} - \beta_{\cdot j}$ denote the estimation error for covariate $j$. Since $k$ is fixed, we treat all factors involving $k$ as constants throughout. We denote $\| \cdot \|_{op}$ as the operator norm of a matrix.

\textbf{Step 1: Basic Inequality from Optimality.}
By the optimality of $\widehat{\beta}_{\cdot j}$, we have
\[
L_{Y_j}(\widehat{\beta}_{\cdot j}; \hat{Z}) + \lambda\sqrt{k}\|\widehat{\beta}_{\cdot j}\|_2 + \frac{\delta}{2}\|\widehat{\beta}_{\cdot j}\|_2^2 
\leq L_{Y_j}(\beta_{\cdot j}; \hat{Z}) + \lambda\sqrt{k}\|\beta_{\cdot j}\|_2 + \frac{\delta}{2}\|\beta_{\cdot j}\|_2^2.
\]

Rearranging, we find:
\begin{align}
L_{Y_j}(\widehat{\beta}_{\cdot j}; \hat{Z}) - L_{Y_j}(\beta_{\cdot j}; \hat{Z}) 
&\leq \lambda\sqrt{k}(\|\beta_{\cdot j}\|_2 - \|\widehat{\beta}_{\cdot j}\|_2) + \frac{\delta}{2}(\|\beta_{\cdot j}\|_2^2 - \|\widehat{\beta}_{\cdot j}\|_2^2) \nonumber \\
&\leq \lambda\sqrt{k}\|\Delta_j\|_2 + \frac{\delta}{2}(\|\beta_{\cdot j}\|_2^2 - \|\widehat{\beta}_{\cdot j}\|_2^2).
\label{eq:basic-ineq-fixed}
\end{align}

\textbf{Step 2: Decomposition via True Design $Z$.}
Adding and subtracting $L_{Y_j}(\widehat{\beta}_{\cdot j}; Z)$ and $L_{Y_j}(\beta_{\cdot j}; Z)$:
\begin{align}
L_{Y_j}(\widehat{\beta}_{\cdot j}; \hat{Z}) - L_{Y_j}(\beta_{\cdot j}; \hat{Z}) 
&= L_{Y_j}(\widehat{\beta}_{\cdot j}; \hat{Z}) -L_{Y_j}(\widehat{\beta}_{\cdot j}; Z) + L_{Y_j}(\widehat{\beta}_{\cdot j}; Z) \\ 
&\quad  - L_{Y_j}(\beta_{\cdot j}; Z) + L_{Y_j}(\beta_{\cdot j}; Z)- L_{Y_j}(\beta_{\cdot j}; \hat{Z}) \;.
\end{align}

These terms can be rearranged: 
\begin{align}
L_{Y_j}(\widehat{\beta}_{\cdot j}; \hat{Z}) - L_{Y_j}(\beta_{\cdot j}; \hat{Z}) 
&= \underbrace{L_{Y_j}(\widehat{\beta}_{\cdot j}; Z) - L_{Y_j}(\beta_{\cdot j}; Z)}_{\text{Loss on true design Z}} \\
&\quad + \underbrace{ L_{Y_j}(\widehat{\beta}_{\cdot j}; \hat{Z}) -L_{Y_j}(\widehat{\beta}_{\cdot j}; Z) + L_{Y_j}(\beta_{\cdot j}; Z)- L_{Y_j}(\beta_{\cdot j}; \hat{Z}) }_{\text{Effect of measurement error}} \;.
\end{align}

The loss on the true design, $Z$, can be separated into an expression involving the score function $\nabla L_{Y_j}(\beta_{\cdot j};Z)$ and the Hessian $\nabla^2 L_{Y_j}(\beta_{\cdot j};Z)$.

First, we can use a Taylor expansion on $L_{Y_j}(\widehat{\beta}_{\cdot j};Z)$:

\begin{align}
L_{Y_j}(\widehat{\beta}_{\cdot j};Z) &= L_{Y_j}(\beta_{\cdot j};Z) + \langle \nabla L_{Y_j}(\beta_{\cdot j};Z), \Delta_j \rangle + \frac{1}{2}\nabla^2 L_{Y_j}(\tilde{\beta}_{\cdot j};Z)
\end{align}
for some $\tilde{\beta}_{\cdot j}$ between $\widehat{\beta}_{\cdot j}$ and $\beta_{\cdot j}$. We can plug this in to the expression for the loss on the true design:

\begin{align}
L_{Y_j}(\widehat{\beta}_{\cdot j}; Z) - L_{Y_j}(\beta_{\cdot j}; Z) &=  \langle \nabla L_{Y_j}(\beta_{\cdot j};Z), \Delta_j \rangle + L_{Y_j}(\widehat{\beta}_{\cdot j}; Z) - L_{Y_j}(\beta_{\cdot j}; Z) - \langle \nabla L_{Y_j}(\beta_{\cdot j};Z), \Delta_j \rangle \;.
\end{align}

Finally, we have

\begin{align}
L_{Y_j}(\widehat{\beta}_{\cdot j}; \hat{Z}) - L_{Y_j}(\beta_{\cdot j}; \hat{Z}) 
&= \underbrace{L_{Y_j}(\widehat{\beta}_{\cdot j}; Z) - L_{Y_j}(\beta_{\cdot j}; Z) - \langle \nabla L_{Y_j}(\beta_{\cdot j}; Z), \Delta_j \rangle}_{\text{(I): Curvature term}} \nonumber \\
&\quad + \underbrace{\langle \nabla L_{Y_j}(\beta_{\cdot j}; Z), \Delta_j \rangle}_{\text{(II): Score term}} \nonumber \\
&\quad + \underbrace{L_{Y_j}(\widehat{\beta}_{\cdot j}; \hat{Z}) - L_{Y_j}(\widehat{\beta}_{\cdot j}; Z) + L_{Y_j}(\beta_{\cdot j}; Z) - L_{Y_j}(\beta_{\cdot j}; \hat{Z})}_{\text{(III): Measurement error bias}}.
\end{align}

\textbf{Step 3: Lower Bound on Curvature Term (I).}
By Assumption \ref{assumption4} (Restricted Strong Convexity):
\[
\text{(I)} \geq \frac{n\kappa}{2}\|\Delta_j\|_2^2.
\]

\textbf{Step 4: Bound on Score Term (II).}
By Assumption \ref{assumption5} (Score Concentration):
\[
|\text{(II)}| = |\langle Z^\top(Y_{\cdot j} - \mu_j(Z)), \Delta_j \rangle| \leq n C_0\sqrt{\frac{k}{n}}\|\Delta_j\|_2 = C_0 \sqrt{nk} \| \Delta_j\|_2 \;.
\]

\textbf{Step 5: Bound on Measurement Error Bias Term (III).}
To understand how to bound term (III), first we replace $\hat{Z} = Z + U$:

\begin{align}
L_{Y_j}(\widehat{\beta}_{\cdot j}; \hat{Z}) - L_{Y_j}(\widehat{\beta}_{\cdot j}; Z) 
& = L_{Y_j}(\widehat{\beta}_{\cdot j}; Z + U) - L_{Y_j}(\widehat{\beta}_{\cdot j}; Z)  = g(\widehat{\beta}_{\cdot j})  \;,
\\
L_{Y_j}(\beta_{\cdot j}; Z) - L_{Y_j}(\beta_{\cdot j}; \hat{Z}) &= L_{Y_j}(\beta_{\cdot j}; Z) - L_{Y_j}(\beta_{\cdot j}; Z + U)  = g(\beta_{\cdot j})  \;.
\end{align}

By mean value theorem, there exists $\tilde{\beta}$ on the line segment between $\beta_{\cdot j}$ and $\widehat{\beta}_{\cdot j}$ such that

\begin{align}
g(\widehat{\beta}_{\cdot j}) - g(\beta_{\cdot j}) = \langle \nabla_\beta g(\tilde{\beta}_{\cdot j}), \Delta \rangle \;.
\end{align}

The gradient with respect to $\beta$ is 

\[\nabla_\beta L(\beta, Z) = Z^\top (\mu(\beta, Z) - Y) \;,\]

where $\mu(\cdot)$ is the logit link. Thus,

\begin{align}
    \nabla_\beta g(\tilde{\beta}_{\cdot j}) &= \nabla_\beta L(\tilde{\beta}, \hat{Z}) - \nabla_\beta L(\tilde{\beta}, Z) \\
    &= \hat{Z}^\top (\mu(\tilde{\beta}, \hat{Z}) - Y) - Z^\top (\mu(\tilde{\beta}, Z) - Y) \\
    &=(Z + U)^\top (\mu(\tilde{\beta}, \hat{Z}) - Y) - Z^\top (\mu(\tilde{\beta}, Z) - Y) \\
    &= U^\top (\mu(\tilde{\beta}, \hat{Z}) - Y) - Z^\top (\mu(\tilde{\beta}, \hat{Z}) - \mu(\tilde{\beta}, Z)) \;,
\end{align}

and we have

\begin{align}
    \langle \nabla_\beta g(\tilde{\beta}_{\cdot j}), \Delta \rangle &= \underbrace{-\langle U^\top (Y - \mu(\tilde{\beta}, \hat{Z})), \Delta \rangle}_{\text{(A)}} - \underbrace{\langle Z^\top (\mu(\tilde{\beta}, \hat{Z}) - \mu(\tilde{\beta}, Z)), \Delta_j \rangle}_{\text{(B)}}  \;.
\end{align}

\textbf{Step 5.1: Bound term (A)}  
For simplicity, we can split terms (A) and (B) to bound them separately.
\begin{align}
(A) = |-\langle U^\top (Y - \mu(\tilde{\beta}, \hat{Z})), \Delta \rangle| & \leq \| U^\top (Y - \mu(\tilde{\beta}, \hat{Z}))\|_2 \; \|\Delta \|_2
\end{align}

Every $Y - \mu(\tilde{\beta}, \hat{Z}) \in [-1, 1]$, so $\|Y - \mu(\tilde{\beta}, \hat{Z})\|_2 \leq \sqrt{n}$. By assumption \ref{assumption6}, we have $\|U\|_\infty \leq \delta_\infty$. Therefore, each column of $U$ has an L2 norm $\leq \sqrt{n} \delta_\infty$ and $\|U\|_{op} \leq \sqrt{n}k \delta_{\infty}$.

Plugging this in to the expression above, we obtain
\begin{align}
    |-\langle U^\top (Y - \mu(\tilde{\beta}, \hat{Z})), \Delta \rangle| & \leq \sqrt{nk} \delta_{\infty} \cdot \sqrt{n} \cdot \|\Delta \|_2 \\
    &= n \sqrt{k}\delta_\infty \|\Delta\|_2  \;.
\end{align}

\textbf{Step 5.2: Bound term (B):}
No focusing on (B):
$(B) = \langle Z^\top (\mu(\tilde{\beta}, \hat{Z}) - \mu(\tilde{\beta}, Z)), \Delta_j \rangle$. 

We can linearize the change in the mean vector $\mu_j$ with respect to the linear predictor:
\begin{align}
    \eta(Z) &= \gamma_j \mathbf{1}_n + Z \tilde{\beta} \;,\\
    \eta(\hat{Z}) &= \gamma_j \mathbf{1}_n + \hat{Z} \tilde{\beta} = \eta(Z) + U\tilde{\beta} \;.
\end{align}

The logit link is differentiable with derivative $\mu'(\xi_i) \in (0, 1/4]$. By the mean value theorem,

\[\mu_j(\tilde{\beta}; \hat{Z}) - \mu_j(\tilde{\beta};Z) = W(U \tilde{B}) \;, \]

where $W = \text{diag}(\mu'_1, \dots \mu'_n)$ and $\|W\|_{op} \leq 1/4$.

By Cauchy-Schwarz, 

\begin{align}
    |(B)| &\leq \|Z^\top W U \|_{op} \|\tilde{\beta} \|_2 \|\Delta_j\|_2  \;.
\end{align}

Given $\| W \|_{op} \leq 1/4 $:

\begin{align}
    \|Z^\top W U \|_{op} \leq \|Z\|_{op} \|W\|_{op} \|U\|_{op} \leq \frac{1}{4}\|Z\|_{op} \|U\|_{op} \;.
\end{align}

We can bound $\|Z\|_{op}$ in terms of $D = \frac{1}{n} Z^\top Z$: for any $v$, $\frac{1}{n}\| Zv \|^2_2 = v^\top D v \leq \nu_{\max}(D) \|v\|^2_2$ :
\[\|Z\|_{op} \leq \sqrt{n \nu_{\max}(D)} \; .\]

With $\|U\|_{op} \leq \sqrt{nk} \delta_\infty$ as above, this can be combined: 

\[\|Z^\top WU\|_{op} \leq \frac{1}{4}\sqrt{n \nu_{\max} (D)} \sqrt{nk}\delta_\infty = \frac{n}{4} \sqrt{k \nu_{\max}(D)} \delta_{\infty} \;.\]

Since $k$ is fixed, $\| \tilde{\beta }\|$ is bounded, and $\| \Delta_j \|_2$ will be small at the target rate, 

\[|(B)| \leq C_B n \delta_\infty \| \Delta_j  \|_2\]

for some constant $C_B$.

Combining terms (A) and (B),
\begin{align}
|\text{(III)}|
&\leq C_3 n\delta_{\infty}\|\Delta_j\|_2 \:,
\end{align}

where $C_3$ is a constant depending on $k$ and $\nu_{\max}(D)$.

\textbf{Step 6: Combining Bounds.}
Substituting Steps 3--5 into the inequality from Step 1 and dividing by $n$:
\begin{align}
\frac{\kappa}{2}\|\Delta_j\|_2^2 &\leq C_0\sqrt{\frac{k}{n}}\|\Delta_j\|_2 + C_3\delta_{\infty}\|\Delta_j\|_2 \nonumber \\
&\quad + \frac{\lambda\sqrt{k}}{n}\|\Delta_j\|_2 + \frac{\delta}{2n}(\|\beta_{\cdot j}\|_2^2 - \|\widehat{\beta}_{\cdot j}\|_2^2).
\end{align}

Using $\|\beta_{\cdot j}\|_2^2 - \|\widehat{\beta}_{\cdot j}\|_2^2 \leq 2\|\beta_{\cdot j}\|_2\|\Delta_j\|_2$, our choice of $\lambda_j = C_1\sqrt{k/n}$, and $\delta = C_2\delta_{\infty}$ (with constants absorbing factors of $k$):
\[
\frac{\kappa}{2}\|\Delta_j\|_2^2 \leq C_4\sqrt{\frac{1}{n}}\|\Delta_j\|_2 + C_5\delta_{\infty}\|\Delta_j\|_2,
\]
where $C_4$ and $C_5$ are constants (depending on $k$, $C_0$, $C_1$, $C_2$, $C_3$, and $\|\beta_{\cdot j}\|_2$).

If $\Delta_j = 0$, we are done. Otherwise, dividing by $\|\Delta_j\|_2$ and solving:
\[
\|\Delta_j\|_2 = O\left(\sqrt{\frac{1}{n}} + \delta_{\infty}\right).
\]

Thus, for any fixed $j$ with appropriate choice of $\lambda_j$:
\[
\|\Delta_j\|_2^2 = O\left(\frac{1}{n} + \delta_{\infty}^2\right).
\]

\textbf{Step 7: Prediction Error on $\hat{Z}$.}
Finally, we bound the prediction error on the estimated design:
\begin{align}
\frac{1}{n}\|\hat{Z}\Delta_j\|_2^2 &= \frac{1}{n}\|(Z + U)\Delta_j\|_2^2 \nonumber \\
&\leq \frac{2}{n}\|Z\Delta_j\|_2^2 + \frac{2}{n}\|U\Delta_j\|_2^2 \nonumber \\
&\leq 2\nu_{\max}(D)\|\Delta_j\|_2^2 + 2C_6\delta_{\infty}^2\|\Delta_j\|_2^2,
\end{align}
where $C_6$ is a constant depending on $k$, using the bound
\[
\frac{1}{n}\|U\Delta_j\|_2^2 \leq \frac{1}{n} \cdot n k \delta_{\infty}^2 \|\Delta_j\|_2^2 = k\delta_{\infty}^2\|\Delta_j\|_2^2 = O(\delta_{\infty}^2\|\Delta_j\|_2^2).
\]

Combining with the bound on $\|\Delta_j\|_2^2$ from Step 6:
\[
\frac{1}{n}\|\hat{Z}\Delta_j\|_2^2 = O_p\left(\frac{1}{n} + \delta_{\infty}^2\right).
\]

This proves the rate in part (a). To prove part (b), we need the following additional steps:

\textbf{Step 7: Uniform Control}
The above analysis holds with high probability for any fixed $j$. To obtain uniform control, note that the score terms behave as independent sub-Gaussian random variables. Therefore, we require $\lambda = C_1 \sqrt{\frac{\log q}{n}}$ to control all $q$ regressions simultaneously. 

With this larger penalty, the uniform bound becomes:

\[
\max_{j = 1, \dots,q} \| \Delta\|_2^2 = O_p \big( \frac{\log q}{n} + \delta^2_{\infty} \big ) \;.
\]

When $\delta_{\infty} = O(n^{-1/2})$ from first-stage estimation:
\[
\frac{1}{n}\|\hat{Z}\Delta_j\|_2^2 = O_p\left(\frac{\log q}{n} + \frac{1}{n}\right) = O_p\left(\frac{\log q}{n}\right),
\]
since $\log q \gg 1$ for high-dimensional settings.

\end{proof}

\end{document}